\documentclass[a4paper,11pt]{amsart}
\usepackage[utf8]{inputenc}
\usepackage{amsaddr}
\usepackage{calrsfs}
\usepackage{graphicx,color}
\usepackage[english]{babel}
\usepackage{graphicx}
\usepackage{amsmath}
\usepackage{epstopdf}
\usepackage{float}
\usepackage{color}

\usepackage{amssymb}

\numberwithin{equation}{section}

\theoremstyle{definition}
\newtheorem{dfn}{Definition}[section]

\theoremstyle{plain}
\newtheorem{thm}{Theorem}[section]

\theoremstyle{definition}

\newtheorem{exa}{Example}[section]

\begin{document}
\title[Valuation of equity warrant for uncertain financial market]
{Valuation of equity warrants for uncertain financial market}

\date{\today}

\author[Shokrollahi]{Foad Shokrollahi}
\address{Department of Mathematics and Statistics, University of Vaasa, P.O. Box 700, FIN-65101 Vaasa, FINLAND}
\email{foad.shokrollahi@uva.fi}

\begin{abstract}
In this paper, within the framework of uncertainty
theory, the valuation of equity warrants is investigated. Different
from the methods of probability theory, the equity warrants pricing problem is solved by using the method of uncertain calculus.
Based on the assumption that the firm price
follows an uncertain differential equation, the equity warrants pricing formula is obtained for uncertain stock model.

\end{abstract}




\maketitle

\section{Introduction}\label{sec:1}

Warrants give the holder the right but not the duty to
buy or sell the underlying asset by a certain date for a certain
price. But this right is not free. Warrant is one kind of
special option and it can been classified lots of types. Warrants
can be divided into American warrants and European
warrants according to the difference of the expiration date.
And it can be divided into call warrants and put warrants
according to the difference of exercise way. While it can be
divided into equity warrants and covered warrants according
to the difference of issuer. Covered warrants are usually
issued by financial institutions and dealers, which do not
raise the company's capital stock after their expiration
dates. Pricing for this kind of warrant is similar to pricing
for ordinary options, therefore, many researchers use
Black–Scholes model \cite{black1973pricing} to price this kind of warrant. But
the equity warrants are usually issued by listed company,
and the underlying capital is the issued stock of its company.
The equity warrants have dilution effect, therefore,
pricing for this kind of warrant is difference to pricing for
the standard European options because the companies'
equity warrants need to issue new stock to meet the request
of warrants' holder at the expiration date. In other words
its pricing cannot completely apply the classics Black–
Scholes formula.

Previous studies of pricing equity warrants are mainly with the method of stochastic finance based on the probability theory, and the firm price are usually assumed to follow some stochastic differential equation \cite{xiao2012valuation, kremer1993warrant,zhang2009equity}. But many empirical investigations showed that the firm value does not behave like randomness, and it is often influenced by the belief degrees of investors since investors usually make their decisions based on the degrees of belief rather than the probabilities. For example, one of the key elements in the Nobel-prize-winning theory of Kahneman and Tversky \cite{tversky1992advances, tversky2016advances} is the finding of probability distortion which showed that decision makers usually make their decisions based on a nonlinear transformation of the probability scale rather than the probability itself, people often overweight small probabilities and underweight large probabilities. We argue that investors' belief degrees play an important
role in decision making for financial practice.

Uncertainty theory was founded by Liu \cite{liu2007uncertainty} in 2007, and it has become a branch of axiomatic mathematics for modeling belief degrees. As a branch of axiomatic mathematics to deal with belief degrees, uncertainty theory will play an important role in financial theory and practice. Liu \cite{liu2009some} initiated the pioneer work of
uncertain finance in 2009. Afterwards, many researchers devoted themselves to study of financial problems by using uncertainty theory. For example, Chen \cite{chen2011american} investigated American option pricing problem and derived the pricing formulae for Liu's uncertain stock model. Chen and Gao \cite{chen2013uncertain} introduced an uncertain term structure
model of interest rate. Besides, based on uncertainty theory, Chen, Liu and Ralescu \cite{chen2013uncertain} proposed an uncertain stock model with periodic dividends and investigated some options pricing for this type of model. Peng and Yao \cite{peng2011new} proposed an uncertain stock model with mean-reverting process, and some option pricing formulae
were investigated on this type of stock model. Yao \cite{yao2012no} gave the no-arbitrage determinant theorems on mean-reverting stock model in uncertain market. Zhang et al. \cite{zhang2016valuation} derived a valuation model of power options under the assumption that the underlying stock price is assumed to follow an uncertain differential equation.

In this paper, the pricing problem of equity warrants is investigated under Liu's uncertain stock model, and the equity warrants pricing formulae are derived under this model. The rest of the paper is organized
as follows: Some preliminary concepts of uncertain processes are recalled in Section \ref{sec:2}. An uncertain
equity warrants model is proposed in Section \ref{sec:3}. Finally, a brief summary is
given in Section \ref{sec:4}.

\section{Preliminary}\label{sec:2}
An uncertain process is essentially a sequence of uncertain variables indexed by time or space. In this section we recall some basic facts about uncertain process.

\begin{dfn}(\cite{liu2008fuzzy}) Let $T$ be an index set and let $(\Gamma, \mathcal{M}, \mathcal{L})  $ be an uncertainty space. An uncertain process is a
measurable function from $T\times (\Gamma, \mathcal{M}, \mathcal{L})$ to the set of real numbers, i.e., for each $t\in T$ and any Borel set $B$ of real numbers, the
set

\begin{eqnarray*}
\{X_t\in B\}=\{\gamma\in\Gamma|X_t(\gamma)\in B\},
\label{eq:1}
\end{eqnarray*}
is an event.
\label{df:1}
\end{dfn}

\begin{dfn}
(\cite{liu2008fuzzy}) The uncertainty distribution $\Phi$ of an uncertain variable $\xi$ is defined by

\begin{eqnarray*}
\Phi (x)=\mathcal{M}\{\xi\leq x\}
\label{eq:2}
\end{eqnarray*}
for any real number x.
\label{df:2}
\end{dfn}

\begin{dfn}(\cite{liu2008fuzzy}) An uncertain variable $\xi$ is called normal if it has a normal uncertainty distribution

\begin{eqnarray*}
\Phi (x)=\left(1+\exp\left(\frac{\pi(e-x)}{\sqrt{3}\sigma}\right)\right)^{-1}
\label{eq:3}
\end{eqnarray*}
denoted by $\mathcal{N}(e, \sigma)$ where $e$ and $\sigma$ are real numbers with $\sigma> 0$.
\label{df:3}
\end{dfn}

\begin{dfn}(\cite{liu2007uncertainty}) Let $\xi$ be an uncertain variable. Then the expected value of $\xi$ is defined by
\begin{eqnarray*}
E[\xi]=\int_0^{+\infty}\mathcal{M}\{\xi\geq r\}dr-\int_{-\infty}^{0}\mathcal{M}\{\xi\leq r\}dr,
\label{eq:4}
\end{eqnarray*}
provided that at least one of the two integrals is finite.
\label{df:4}
\end{dfn}

\begin{thm}
(\cite{liu2007uncertainty}) Let $\xi$ be an uncertain variable with uncertainty distribution $\Phi$. If the expected value exists, then
\begin{eqnarray*}
E[\xi]=\int_0^{+\infty}(1-\Phi(x))dx-\int_{-\infty}^{0}\Phi(x)dx.
\label{eq:4}
\end{eqnarray*}
\label{thm:1}
\end{thm}

\begin{dfn} (\cite{liu2010uncertainty}) An uncertainty distribution $\Phi(x)$ is said to be regular if it is a continuous and strictly increasing function with respect to $x$ at which $0 <\Phi(x)< 1$, and
\begin{eqnarray*}
\lim_{x\rightarrow -\infty}\Phi(x)=0,\quad \lim_{x\rightarrow +\infty}\Phi(x)=1.
\label{eq:5}
\end{eqnarray*}
\label{df:5}
\end{dfn}

\begin{thm}
(\cite{liu2010uncertainty}) Let $\xi$ be an uncertain variable with regular uncertainty distribution $\Phi$. Then
\begin{eqnarray*}
E[\xi]=\int_0^1\Phi^{-1}(\alpha))d\alpha.
\label{eq:6}
\end{eqnarray*}
\label{thm:2}
\end{thm}

\begin{dfn} (\cite{liu2009some}) An uncertain process $C_t$ is said to be a canonical Liu process if

\begin{itemize}

\item[(i)] $C_0=0$ and almost all sample paths are Lipschitz continuous;
\item[(ii)] $C_t$ has stationary and independent increments;
\item[(ii)]    every increment $C_{s+t}-C_s$ is a normal uncertain variable with expected value $0$ and variance $t^2$.
\end{itemize}
\label{df:6}
\end{dfn}

\begin{dfn} (\cite{chen2013liu}) Let $C_t$ be a canonical Liu process and let $Z_t$ be an uncertain process. If there exist uncertain processes $\mu_t$ and $\sigma_t$ such that

\begin{eqnarray*}
Z_t=Z_0+\int_0^t\mu_s ds+\int_0^t\sigma_s dC_s,
\label{eq:7}
\end{eqnarray*}
for any $t\geq 0$, then $Z_t$ is called a Liu process with drift $\mu_t$ and diffusion $\sigma_t$. Furthermore, $Z_t$ has an uncertain differential
\begin{eqnarray*}
dZ_t=\mu_t dt+\sigma_t dC_t.
\label{eq:8}
\end{eqnarray*}
\label{df:7}
\end{dfn}

\begin{dfn} (\cite{liu2008fuzzy}) Suppose $C_t$ is a canonical Liu process, and $f$ and $g$ are two functions. Then

\begin{eqnarray*}
dX_t=f(t, X_t) dt+g(t, X_t) dC_t,
\label{eq:9}
\end{eqnarray*}
is called an uncertain differential equation.
\label{df:8}
\end{dfn}

\begin{dfn} (\cite{yao2013numerical}) Let $\alpha$ be a number with $0<\alpha< 1$. An uncertain differential equation

\begin{eqnarray*}
dX_t=f(t, X_t) dt+g(t, X_t) dC_t,
\label{eq:10}
\end{eqnarray*}
is said to have an $\alpha$-path $X_t^\alpha$ if it solves the corresponding ordinary differential equation
\begin{eqnarray*}
dX_t^\alpha=f(t, X_t^\alpha) dt+|g(t, X_t^\alpha)| \Phi^{-1}(\alpha)dt,
\label{eq:11}
\end{eqnarray*}
where $\Phi^{-1}(\alpha)$ is the inverse standard normal uncertainty distribution, i.e.,
\begin{eqnarray*}
\Phi^{-1}(\alpha)=\frac{\sqrt{3}}{\pi}\ln\frac{\alpha}{\alpha-1}.
\label{eq:12}
\end{eqnarray*}
\label{df:9}
\end{dfn}

\begin{thm}
(\cite{yao2013numerical}) Let $X_t$ and $X_t^\alpha$ be the solution and $\alpha$-path of the uncertain differential equation
\begin{eqnarray*}
dX_t=f(t, X_t) dt+g(t, X_t) dC_t,
\label{eq:13}
\end{eqnarray*}
respectively. Then
\begin{eqnarray*}
\mathcal{M}\{X_t\leq X_t^\alpha, \forall t\}&=&\alpha\nonumber\\
\mathcal{M}\{X_t> X_t^\alpha, \forall t\}&=&1-\alpha.
\label{eq:14}
\end{eqnarray*}
\label{thm:3}
\end{thm}

\begin{thm}
(\cite{yao2013numerical}) Let $X_t$ and $X_t^\alpha$ be the solution and $\alpha$-path of the uncertain differential equation
\begin{eqnarray*}
dX_t=f(t, X_t) dt+g(t, X_t) dC_t,
\label{eq:15}
\end{eqnarray*}
respectively. Then the solution $X_t$ has an inverse uncertainty distribution
\begin{eqnarray*}
\Psi_t^{-1}(\alpha)=X_t^\alpha.
\label{eq:16}
\end{eqnarray*}
\label{thm:4}
\end{thm}

\begin{thm}
(\cite{yao2013numerical}) Let $X_t$ and $X_t^\alpha$ be the solution and $\alpha$-path of the uncertain differential equation
\begin{eqnarray*}
dX_t=f(t, X_t) dt+g(t, X_t) dC_t,
\label{eq:17}
\end{eqnarray*}
respectively. Then for any monotone (increasing or decreasing) function $I$, we have
\begin{eqnarray*}
E[I(X_t)]=\int_0^1I(X_t^\alpha) d\alpha.
\label{eq:18}
\end{eqnarray*}
\label{thm:5}
\end{thm}

\section{The pricing model}\label{sec:3}
Given a uncertainty space $(\Gamma, \mathcal{M}, \mathcal{L})$,
we will assume ideal conditions in the market for the firm's
value and for the equity warrants:

\begin{itemize}

\item[(i)] There are no transaction costs or taxes and all securities are perfectly divisible;
\item[(ii)] Dividends are not paid during the lifetime of the outstanding warrants, and the sequential exercise of the warrants is not optimal for warrant holders;
\item[(iii)] The warrant-issuing firm is an equity firm with no outstanding debt;
\item[(iv)] The total equity value of the firm, during the lifetime of the outstanding warrants, $V_t$, follows under uncertainty measure $\mathcal{M}$:
\begin{eqnarray}
dV_t=\mu V_tdt+\sigma V_tdC_t,
\label{eq:19}
\end{eqnarray}
where $\mu$ and $\sigma$ are respectively the expected return and volatility rates, $C_t$ is a is a canonical Liu process and $V_t$ be the asset
value of the firm at time $t$.
\end{itemize}

Now, we consider a uncertain stock model

\begin{equation}
\begin{cases}
dP_t=rP_tdt\\
dV_t=\mu V_tdt+\sigma V_tdC_t,
\end{cases}
\label{eq:20}
\end{equation}

where $r$ represents the interest rate.

It follows from the Eq. (\ref{eq:20}) that the firm's value is

\begin{eqnarray}
V_t=V_0e^{\mu t+\sigma C_t},\quad 0\leq t\leq T,
\label{eq:21}
\end{eqnarray}

whose inverse uncertainty distribution is

\begin{eqnarray*}
\Phi^{-1}(\alpha)=V_0\exp\left\{\mu t+\frac{\sigma t \sqrt{3} }{\pi}\ln\frac{\alpha}{\alpha-1}\right\}.
\label{eq:22}
\end{eqnarray*}

In the case of equity warrants that the firm has $N$ shares of common stock and $M$ shares of equity warrants outstanding. Each warrant
entitles the owner to receive $k$ shares of stock at time $T$ upon payment of $J$, the payoff of equity warrants is given by $\frac{1}{N+Mk}[kV_T-NJ]^+$, where $V_T$ be the value of the firm's assets at time $T$. Considering the time value of money resulted from the bond, the present value of this payoff is
\begin{eqnarray*}
\frac{e^{-r(T-t)}}{N+Mk}[kV_T-NJ]^+.
\label{eq:23}
\end{eqnarray*}

\begin{dfn} Assume that there is a firm financed by $N$ shares of stock and $M$ shares of equity warrants. Each warrant gives
the holder the right to buy $k$ shares of stock at time $t = T$ in exchange for the payment of an amount $J$. Let $V_t$ be the asset
value of the firm at time $t$. Then the equity warrant price is

\begin{eqnarray*}
f_w=\frac{e^{-r(T-t)}}{N+Mk}E\left[(kV_T-NJ)^+\right].
\label{eq:24}
\end{eqnarray*}

\label{df:10}
\end{dfn}

\begin{thm}

Assume that the firm has $N$ shares of stock and $M$ shares of equity warrants outstanding. Each warrant gives
the holder the right to buy $k$ shares of stock at time $t = T$ upon payment of $J$. Let $V_t$ be the asset
value of the firm at time $t$. Let $S_t$ be the value of the stock and $\sigma_s$ be its volatility. Then the price of an equity warrant at time $t$ is given by

\begin{eqnarray*}
f_w=\frac{e^{-r(T-t)}}{N+Mk}\int_0^1\left[kV_t\exp\left\{\mu(T-t)+\frac{\sigma\sqrt{3}(T-t)}{\pi}\ln\frac{\alpha}{1-\alpha}\right\}-NJ\right]^+d\alpha,
\label{eq:25}
\end{eqnarray*}

where the optimal solutions $\sigma^*$ and $V_t^*$ satisfy the following system of nonlinear equations
\begin{equation}
\begin{cases}
NS_t=V_t-Mf_w\\
\sigma_s=\frac{\sigma V_t}{S_t}\left(\frac{1}{N}-M\frac{\partial f_w}{\partial V_t}\right).
\end{cases}
\label{eq:26}
\end{equation}
\label{thm:6}
\end{thm}

\begin{proof}
Solving the ordinary differential equation
\begin{eqnarray*}
dV_t^\alpha=\mu V_t^\alpha dt+\sigma V_t^\alpha \Phi^{-1}(\alpha)dt,
\label{eq:27}
\end{eqnarray*}

where $0<\alpha<1$ and $\Phi^{-1}(\alpha)$ is the inverse standard normal uncertainty distribution, we have

\begin{eqnarray*}
V_t^\alpha=V_0\exp\left\{\mu t+\sigma\Phi^{-1}(\alpha)t\right\}.
\label{eq:28}
\end{eqnarray*}
That means that the uncertain differential equation $dV_t=\mu V_tdt+\sigma V_tdC_t$ has an $\alpha$-path
\begin{eqnarray*}
V_t^\alpha&=&V_0\exp\left\{\mu t+\sigma\Phi^{-1}(\alpha)t\right\}\\
&=&V_0\exp\left\{\mu t+\frac{\sigma\sqrt{3}}{\pi}\ln\frac{\alpha}{\alpha-1}\right\}.
\label{eq:29}
\end{eqnarray*}
Since $I(x )=\frac{e^{-r(T-t)}}{N+Mk}[kV_T-NJ]^+$ is an increasing function, it follows from Theorem \ref{thm:5} and Definition \ref{df:10} that the equity warrant price is

\begin{eqnarray*}
f_w&=&E\left[I(V_T)\right]=\int_0^1I(V_T^\alpha)d\alpha\\
&=&\frac{e^{-r(T-t)}}{N+Mk}\int_0^1\left[kV_T^\alpha-NJ\right]^+d\alpha\\
&=&\frac{e^{-r(T-t)}}{N+Mk}\int_0^1\left[kV_t\exp\left\{\mu(T-t)+\frac{\sigma\sqrt{3}(T-t)}{\pi}\ln\frac{\alpha}{1-\alpha}\right\}-NJ\right]^+d\alpha.
\label{eq:30}
\end{eqnarray*}
It is shown that the warrant pricing formula mentioned above depends on $V_t$ and $\sigma$, which are unobservable. To obtain
a pricing formula using observable values, we will make use of the following result.

Let $\beta$ be the stock's elasticity, which gives the percentage change in the stock's value for a percentage change in the
firm's value. Then from a standard result in option pricing theory, we have
\begin{eqnarray}
\beta=\frac{\sigma_s}{\sigma}=\frac{V_t\partial S_t}{S_t\partial V_t}.
\label{eq:31}
\end{eqnarray}
From assumption (iii), we obtain $V_t=NS_t+Mf_w$. Consequently, we have
\begin{eqnarray}
\frac{\partial S_t}{\partial V_t}=\frac{1}{N}-M\frac{\partial f_w}{\partial V_t}.
\label{eq:32}
\end{eqnarray}
Now, from (\ref{eq:31}) and (\ref{eq:32}), it follows that

\begin{eqnarray}
\sigma_s=\frac{\sigma V_t}{S_t}\left[\frac{1}{N}-M\frac{\partial f_w}{\partial V_t}\right].
\label{eq:33}
\end{eqnarray}
\end{proof}

\begin{thm}
The nonlinear system (\ref{eq:26}) has a solution $(\sigma^*, V_t^*)\in (0, +\infty)\times (0, +\infty)$.
\label{thm:7}
\end{thm}

\begin{proof}
First, it is clear that for any $\sigma\in (0, +\infty)$, there exists an unique $V_t\in (0, +\infty)$ which satisfies

\begin{eqnarray*}
NS_t=V_t-Mf_w.
\label{eq:34}
\end{eqnarray*}
Define a map $g:\sigma\rightarrow V_t$, which is given by an implicit function
\begin{eqnarray*}
G(\sigma, V_t)=V_t-Mf_w-NS_t.
\label{eq:35}
\end{eqnarray*}

The function $g:\sigma\mapsto V_t$ is increasing since the following inequality holds:

\begin{eqnarray*}
\frac{dV_t}{d\sigma}=-\frac{\partial G/\partial \sigma}{\partial G/\partial V_t}=\frac{M\frac{\partial f_w}{\partial \sigma}}{1-M\frac{\partial f_w}{\partial V_t}}>0.
\label{eq:36}
\end{eqnarray*}
The inequality holds true because the function $f_w$ is an increasing function of $\sigma$.

Second, it is obvious that for any $\sigma\in (0, +\infty)$, there exists an unique $V_t(\sigma)\in (0, +\infty)$, which satisfies

\begin{eqnarray*}
\sigma_s=\frac{\sigma V_t}{S_t}\left[\frac{1}{N}-M\frac{\partial f_w}{\partial V_t}\right].
\label{eq:37}
\end{eqnarray*}

Define a map $h:\sigma\mapsto V_t$, which is given by an implicit function

\begin{eqnarray*}
H(\sigma, V_t)=\frac{\sigma V_t}{S_t}\left[\frac{1}{N}-M\frac{\partial f_w}{\partial V_t}\right]-\sigma_s.
\label{eq:38}
\end{eqnarray*}
Function $h$ is strictly continuous in $V_t$ for all positive $\sigma$. Moreover, for all $\sigma>0$, $\lim_{V_t\rightarrow 0}h(\sigma, V_t)=0$ and $\lim_{V_t\rightarrow +\infty}h(\sigma, V_t)=+\infty$.
Thus we have
\begin{itemize}

\item[(1)] $g$ is one to one, continuous and strictly increasing;
\item[(2)] $h$ is continuous and attains any value in $(0, +\infty)$.
\end{itemize}

Hence the intersection of $g$ and $h$ exits. This completed the proof.

\end{proof}
\begin{exa}
Let $N=50, T-t=3, M=100, k=1, S_t=100, \sigma_s=0.04, J= 50, r=0.04, \mu=0.02$. Then based on approximations $V_t\approx NS_t$ and $\sigma\approx\sigma_s$, the value of equity warrant is
\begin{eqnarray*}
f_w=16.83
\label{eq:39}
\end{eqnarray*}
\end{exa}
\section{Conclusion}\label{sec:4}

The equity warrants pricing was investigated within the framework of uncertainty theory in this paper. Based on the assumption that the firm value follows an uncertain differential equation, the formulas of equity warrants for Liu's uncertain stock model were derived with the method of uncertain calculus.

\bibliographystyle{siam}
\bibliography{../../reference}

\begin{thebibliography}{10}

\bibitem{black1973pricing}
{\sc F.~Black and M.~Scholes}, {\em The pricing of options and corporate
  liabilities}, Journal of {P}olitical {E}conomy, 81 (1973), pp.~637--654.

\bibitem{chen2011american}
{\sc X.~Chen}, {\em American option pricing formula for uncertain financial
  market}, International {J}ournal of {O}perations {R}esearch, 8 (2011),
  pp.~32--37.

\bibitem{chen2013uncertain}
{\sc X.~Chen and J.~Gao}, {\em Uncertain term structure model of interest
  rate}, Soft {C}omputing, 17 (2013), pp.~597--604.

\bibitem{chen2013liu}
{\sc X.~Chen and D.~A. Ralescu}, {\em Liu process and uncertain calculus},
  Journal of {U}ncertainty {A}nalysis and {A}pplications, 1 (2013), p.~3.

\bibitem{kremer1993warrant}
{\sc J.~W. Kremer and R.~L. Roenfeldt}, {\em Warrant pricing: jump-diffusion
  vs. {B}lack-{S}choles}, Journal of {F}inancial and {Q}uantitative {A}nalysis,
  28 (1993), pp.~255--272.

\bibitem{liu2007uncertainty}
{\sc B.~Liu}, {\em Uncertainty theory}, Springer, 2007.

\bibitem{liu2008fuzzy}
\leavevmode\vrule height 2pt depth -1.6pt width 23pt, {\em Fuzzy process,
  hybrid process and uncertain process}, Journal of {U}ncertain {S}ystems, 2
  (2008), pp.~3--16.

\bibitem{liu2009some}
\leavevmode\vrule height 2pt depth -1.6pt width 23pt, {\em Some research
  problems in uncertainty theory}, Journal of {U}ncertain {S}ystems, 3 (2009),
  pp.~3--10.

\bibitem{liu2010uncertainty}
{\sc B.~Liu}, {\em Uncertainty theory: {A} {B}ranch of {M}athematics for
  {M}odeling {H}uman {U}ncertain}, 2010.

\bibitem{peng2011new}
{\sc J.~Peng and K.~Yao}, {\em A new option pricing model for stocks in
  uncertainty markets}, International {J}ournal of {O}perations {R}esearch, 8
  (2011), pp.~18--26.

\bibitem{tversky1992advances}
{\sc A.~Tversky and D.~Kahneman}, {\em Advances in prospect theory:
  {C}umulative representation of uncertainty}, Journal of {R}isk and
  {U}ncertainty, 5 (1992), pp.~297--323.

\bibitem{tversky2016advances}
\leavevmode\vrule height 2pt depth -1.6pt width 23pt, {\em Advances in prospect
  theory: {C}umulative representation of uncertainty}, in Readings in {F}ormal
  {E}pistemology, Springer, 2016, pp.~493--519.

\bibitem{xiao2012valuation}
{\sc W.~Xiao, W.~Zhang, W.~Xu, and X.~Zhang}, {\em The valuation of equity
  warrants in a fractional {B}rownian environment}, Physica {A}: {S}tatistical
  {M}echanics and its {A}pplications, 391 (2012), pp.~1742--1752.

\bibitem{yao2012no}
{\sc K.~Yao}, {\em No-arbitrage determinant theorems on mean-reverting stock
  model in uncertain market}, Knowledge-Based Systems, 35 (2012), pp.~259--263.

\bibitem{yao2013numerical}
{\sc K.~Yao and X.~Chen}, {\em A numerical method for solving uncertain
  differential equations}, Journal of {I}ntelligent \& {F}uzzy Systems, 25
  (2013), pp.~825--832.

\bibitem{zhang2009equity}
{\sc W.-G. Zhang, W.-L. Xiao, and C.-X. He}, {\em Equity warrants pricing model
  under fractional {B}rownian motion and an empirical study}, Expert {S}ystems
  with {A}pplications, 36 (2009), pp.~3056--3065.

\bibitem{zhang2016valuation}
{\sc Z.~Zhang, W.~Liu, and Y.~Sheng}, {\em Valuation of power option for
  uncertain financial market}, Applied {M}athematics and {C}omputation, 286
  (2016), pp.~257--264.

\end{thebibliography}

\end{document}